\documentclass{article}

\usepackage[english]{babel}

\usepackage[letterpaper,top=2cm,bottom=2cm,left=3cm,right=3cm,marginparwidth=1.75cm]{geometry}

\usepackage{amsmath}
\usepackage{amsfonts}
\usepackage{amsthm}
\usepackage{graphicx}
\usepackage{subfigure}
\usepackage{multicol}
\usepackage{array}
\usepackage{booktabs}
\usepackage{makecell}
\usepackage[colorlinks=true, allcolors=blue]{hyperref}

\usepackage{algorithm}
\usepackage[algo2e]{algorithm2e}
\usepackage{subcaption}
\usepackage{tikz,pgfplots}
\usepackage{tabularx}

\newtheorem{theorem}{Theorem}
\newtheorem{definition}{Definition}[section]

\title{Practical Secure Inference Algorithm for Fine-tuned Large Language Model Based on Fully Homomorphic Encryption}
\author{Zhang Ruoyan, Zheng Zhongxiang\thanks{Corresponding author at: Communication University of China, Beijing, China. E-mail address: zhengzx@cuc.edu.cn}, Bao Wankang}

\begin{document}
\maketitle

\begin{abstract}
Large language models(LLMs) are currently at the forefront of the machine learning field, which show a broad application prospect but at the same time expose some risks of privacy leakage. Both the training datasets and the user's data inputted during interaction are facing security issues, which need to be solved urgently before their further development.

To address this problem, we combined privacy-preserving techniques such as Fully Homomorphic Encryption(FHE) and provable security theory with Parameter-Efficient Fine-Tuning(PEFT) to propose an efficient and secure inference scheme for LLMs that protects both the user-side's input and the server-side's private parameters. More specially, we focus on pre-trained LLMs which rely on open-sourced base model and then fine-tuned with the private datasets by LoRA. This is a popular road-map for Vertical Domain Large Models such as LawGPT and BenTsao.

To achieve this efficient and secure inference LLM scheme, we use two key technologies that are summarized below.
\begin{itemize}
    \item Firstly, we divide the whole model into two parts, denoted as the public part and the private part. The weights of public part are publicly accessible(e.g. the open-sourced base model) while the private part needs to be protected(e.g. the LoRA matrices). Then the public part is deployed on the client side, and the server keeps the private part. In this way, the overhead brought by computing on private data can be greatly reduced. 
 
    \item Secondly, we propose a general method to transform a linear layer into another one which provides security against model extraction attacks and preserves its original functionality, which denoted as Private Linear Layer(PLL). Then we use this method on the LoRA matrices of the server-side where PLL changes the computation of the LoRA matrices in a way that accomplishes correct inference and makes sure that the server protects their private weights without restricting the user's input. We also show that the difficulty of performing model extraction attacks for PLL can be reduced to the well-known hard problem Learning with Errors(LWE). Combing this method with FHE, we can get an inference algorithm for fine-tuned LLM which protects user's input and the server's private weights at the same time.
\end{itemize}
In this paper, we use the open-source model ChatGLM2-6B as the base model which is fine-tuned by LoRA. Experimental results show the inference efficiency of our scheme reaches 1.61s/token which displays that the scheme has good practicality.
\end{abstract}

\section{Introduction}
ChatGPT sparked heated debates as soon as it was launched in 2022, and has been in the public eye so far because of its powerful language understanding and generation with coherent and logical contextual semantics. It is now a hot topic for researchers and a direction of development for companies. And it has long been integrated into people's lives as a tool that has changed modern human-computer interaction. Such large language models(LLMs) have shown a broad application prospect, however, they also expose the risk of privacy leakage.

For accurate responses to common sense questions, all machine learning models display varying degrees of memorization phenomena\cite{Memory}. LLMs likewise have a certain tendency to memorize the underlying training data directly. When given the top text of a training dataset, it has a high probability of directly outputting the followed text of that data in the training dataset, e.g. Carlini et al. used this method to extract 600 pieces of training data from GPT-2\cite{Carlini}. Their demonstration reveals that the training data is less secure, and that this memorization enables an attacker to extract the underlying training data by posing questions wherever they have access to the model. According to the privacy policy of OpenAI\cite{openai_privacy}, the company that developed ChatGPT, they collect personal information that is included in the inputs, file uploads, or feedbacks and use these information to improve their services and conduct research. This suggests that users' input data may also enter the training set and face the extraction problem described above in practice. Because current terminal devices still cannot easily deploy LLMs locally, private inputs have to be transmitted to the server when users interact. These private data can easily be intercepted by others, recorded by the server and used for training purposes if it is not encrypted.

For the purpose of ensuring the security of the input data, privacy computing techniques for small deep learning models such as Deep Neural Networks(DNNs) and Convolutional Neural Networks(CNNs) have been extensively studied. For example, in 2018, Badawi et al. proposed a CNNs based entirely on Fully Homomorphic Encryption(FHE) that is able to homomorphically classify encrypted images\cite{FHECNN}. Their solution achieved sufficient security level and reasonable classification accuracy. However, such research on large models is currently in its initial stages. Nowadays, most of large models are based on the transformer model\cite{attention}, which is a neural network architecture based on the self-attention, and is widely used in natural language processing tasks. Nevertheless, there are a few existing studies on privacy-preserving techniques for the transformer model. Chen et al. proposed a privacy-preserving inference method THE-X for approximation computation of transformer models based on FHE in 2022\cite{THE-X}. They used a series of polynomial approximation nonlinear operations to implement the inference process. While approximating the activation function for inference is difficult to achieve the correct output and good efficiency\cite{FHECNN}. As the model size becomes larger and more approximations are obtained, the efficiency and accuracy of the model becomes more difficult to ensure. Private transformer inference systems such as Iron\cite{Iron}, BumbleBee\cite{BumbleBee}, and CipherGPT\cite{CipherGPT} were proposed in 2022 and 2023. These systems propose secure two-party protocols for complex non-linear functions, including softmax, GELU, and LayerNorm. However, the parameter scale of the transformer model they used and the efficiency cannot meet practical requirements. In September 2023, Dong et al. proposed an end-to-end secure transformer inference framework PUMA for large transformer models under Secure Multi-Party Computation based on top of SecretFlow-SPU\cite{PUMA}, and realized the lossless inference of pre-trained large models under MPC by designing high-precision operators such as softmax and GELU. The inference efficiency of PUMA on LLAMA-7B reaches 200s/token.

Another common method to improve the efficiency of ciphertext computation when using homomorphic encryption is to circumvent complex ciphertext computation with transmission. The complex computation is passed back to the data owner to be completed in plaintext state. Using this idea, Lam et al. proposed a hybrid PE-NN model based on FHE in 2023\cite{PE-NN}, where a pre-trained open-source image prediction CNN model is followed by an additional linear network. Then the CNN model is deployed on the client side and the linear network on the server side. The client uses the local open-source network to compute the input, then encrypts and transmits the encrypted data to the server. And the server returns an encrypted predicted output through the private network. This method guarantees that user data is transmitted in ciphertext while obviating the need for nonlinear operations to be performed using private computation via FHE. By substituting the computationally intensive ciphertext evaluations on the server side with low-cost plaintext computations on the client side, this approach effectively reduces the overhead associated with private computations. However, the model's capability is bounded by the ability of the CNN, because the private part is made up of a linear network which does not improve the ability of the CNN model to deal with problems. Besides, a linear network is weak in terms of security against Model Extraction Attacks, an attacker can calculate the model weights by modifying the input data and observing changes in the output results, which poses a significant security risk.

To avoid these problems, we focus on the Low-Rank Adaptation(LoRA)\cite{LoRA} technique for fine-tuning LLMs. In order to flexibly adapt models to the specific task requirements of a vertical domain, fine-tuning on pre-trained models has become a paradigm for dealing with natural language processing tasks, which can lead to huge performance gains on specific tasks\cite{QLoRA}\cite{lee2024platypusquickcheappowerful}\cite{chaves2024txllmlargelanguagemodel}. However, as models grow in size, fine-tuning all model parameters on standard consumer hardware becomes impractical. To address this problem, an efficient fine-tuning technique called Parameter-Efficient Fine-Tuning(PEFT) has become popular. PEFT can fix most of the pre-trained model parameters and fine-tune only a small number or additional model parameters. Therefore, PEFT significantly reduces computational and storage costs while achieving performance comparable to full parameter fine-tuning. The common local fine-tuning techniques are Prefix Tuning\cite{PrefixTuning}, Prompt Tuning\cite{PromptTuning}, Low-Rank Adaptation(LoRA)\cite{LoRA}, etc. Among them, LoRA fine-tuning technique\cite{LoRA} is one of the most versatile and at the same time the most effective fine-tuning methods. More specifically, LoRA freezes the pre-trained model weights and injects trainable bypass matrices into each layer of the Transform architecture. These bypass matrices are used to simulate full-parameter fine-tuning based on the intrinsic low-rank nature of large models, and the use of two low-rank matrices in the bypass matrices reduces the number of parameters that need to be tuned in the model from $d*d$ to $2*r*d\,(r\ll d)$ without changing the dimension of the output data. Finally when saving the fine-tuning weights, one only needs to save the portion of the low-rank matrices. In this way, LLMs with LoRA will not change the architecture of the original model, and will not produce a delay in inference.

Combining private technologies with LoRA, we propose an efficient and secure inference scheme for LLMs with LoRA based on FHE. We deploy the base LLM on the client and the low-rank matrix obtained by fine-tuning on the server. In the inference process, the user encrypts the data when it needs to pass through the low-rank matrix section and sends it to the server, then the server completes ciphertext computation before returning the ciphertext result back. Finally, the user decrypts the result to finish the inference. This method can greatly reduce the overhead brought by computing on the private data and improve the inference efficiency.

In this process, an important security issue is that how to protect the private weights on the server side because the low-rank matrices is a simple linear layer, and the parameters of this part can be easily computed by an attacker using Model Extraction Attacks. To deal with the problem, we introduce a general transform method which can be applied to any linear layer to gain an extra protection against Model Extraction Attacks known as Private Linear Layer(PLL). We show that PLL can improve the security of the original linear layer and preserves its original functionality. We also show that the difficulty of performing Model Extraction Attacks for PLL can be reduced to a well-known hard problem LWE. 

By combining these techniques, our solution can thoroughly minimize ciphertext computations to achieve high efficiency and simultaneously protect the security of the user's input of the client side and the private weights of the server side. The article is organized as follows. 
In Section 2, we introduce the mathematical notations and related techniques utilized in this paper. Then we propose the architecture of our scheme with its security analysis in  Section 3. And in Section 4, we give the experimental results of our scheme including a comparison with former methods. Finally, we summarize our contributions and discuss potential future work in Section 5.



\section{Background}
\subsection{Notation}
The following is a description of the symbols to be used in this paper:
\begin{itemize}
\item $\mathbb{R}$. The set of real number. The matrix with dimension $m\times n$ belonging to set $\mathbb{R}$ is denoted as $\mathbb{R}^{m\times n}$.
\item $x\sim D$. Sample from distribution $D$ yields $x$. 
\item $\mathbb{Z}_q^m$. The integer matrix ring with dimension $m$, where addition, subtraction and multiplication are performed modulo $q$.
\item $R=\mathbb{Z}[X]/(X^N+1)$. For a power-of-two integer $N$, the cyclotomic polynomial ring of dimension $N$.
\item $x\gets F$. Assigning the result of function $F$ to $x$.
\item $Enc(m)$. Encrypt plaintext $m$.
\item $Dec(ct)$. Decrypt ciphertext $ct$.
\item $Add(ct_1,ct_2)$. Return the sum of $ct_1,ct_2$ in ciphertext.
\item $CMult(ct;c)$. For ciphertext $ct_1=Enc[a_1,a_2,\cdots,a_n]$ and plaintext $c=[b_1,b_2,\cdots,b_n]$, return the ciphertext $ct_2=Enc[a_1*b_1,a_2*b_2,\cdots,a_n*b_n]$.
\item $Mult(ct_1;ct_2)$. For ciphertext $ct_1=Enc[a_1,a_2,\cdots,a_n]$ and $ct_2=Enc[b_1,b_2,\cdots,b_n]$, return the ciphertext $ct_3=Enc[a_1*b_1,a_2*b_2,\cdots,a_n*b_n]$.
\item $Rescale(ct;p)$. For a ciphertext $ct$ and an integer $p$, reduce the level of ciphertext while reducing noise.
\item $Rotate(ct;r)$. Return a ciphertext encrypting the rotated plaintext vector of $ct$ by $r$ positions.
\end{itemize}

\subsection{Transformer Model}
The Seq2Seq model\cite{S2Sequnce} is an important text generation model in natural language processing(NLP). This model is employed to handle variable-length input and output sequences. The Seq2Seq models have an encoder-decoder structure. The encoder is used to encode the information of the input sequence, which encodes the information contained in the input sequence of arbitrary length into an information vector. The decoder is used to decode the information vector and generate the output sequence.

In the Seq2Seq model based on the traditional Recurrent Neural Network (RNN), both the encoder and the decoder are generally RNNs with the same structure. The encoder's RNN processes the input sequences which compresses the input sequence information into a vector. The last state of the encoder contains the information of the entire input sequence and serves as the initial state for the decoder's RNN. Because the computation of this next layer requires the state of the previous layer, the RNN-based model cannot be parallelized. Besides, Cho et al.\cite{RNNlong} shows that when the input sequence is very long, the encoder will more or less forget some of the information in the input sequence, and thus the decoder is obviously not able to generate a correct result. This indicates that RNN-based model is not suitable for handling long sequences of sentences.

To address this problem, Vaswani et al.\cite{attention} proposed the transformer architecture in 2017. The transformer architecture is a neural network architecture consisting of self-attention layers and feed-forward network which is now widely used in natural language processing tasks, including LLMs.

The encoder of transformer contains two sub-layers, namely the multi-head attention mechanism and the feed-forward network, while the decoder contains the third sub-layers known as the masked multi-head attention. Multi-head attention mechanism is formed by combining multiple self-attention, followed by an Add \& Norm layer. In Add \& Norm layers, Add stands for Residual Connection to prevent network degradation, and Norm stands for Layer Normalization, which is used to normalize the activation values of each layer. 

\textbf{Self-attention. }Self-attention is a variant of the attention mechanism that reduces the dependence on external information and is good at capturing the dependencies between different positions in the sequence. The structure is shown in Fig.\ref{fig:selfattention} .

The input of self-attention is denoted as matrix $X$.
And $Q, K, V$ are obtained by linearly transforming the matrix $X$ as shown in Fig.\ref{fig:QKV}.  The trainable parameter matrices are $W^Q$, $W^K$, and $W^V$ and the input matrix $X$ is multiplied with the three matrix parameters respectively to obtain $Q$, $K$, and $V$. Among them, $Q$ denotes query vectors, $K$ denotes key vectors, and $V$ denotes value vectors. The attention matrix is then computed using $Q$, $K$ and $V$ with the following formula, where $M$ can be considered as the deviation matrix. $$Attention(Q,K,V) = softmax(Q·K^T\cdot M)·V$$

\begin{figure}[H]
\centering
\subfigure[Self-attention]{
\begin{minipage}[t]{0.3\textwidth}
\centering
\includegraphics[scale=0.5]{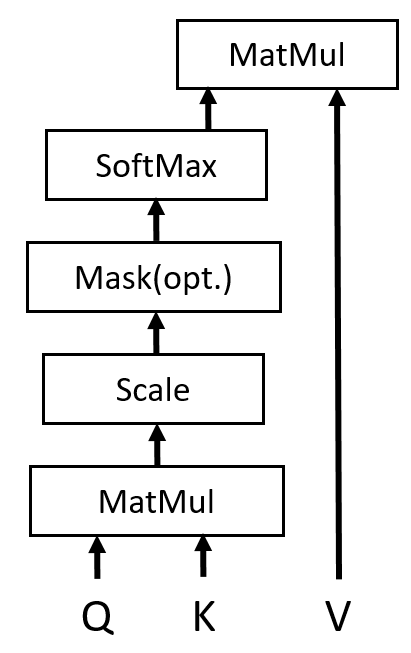}
\label{fig:selfattention}
\end{minipage}
}
\hfill
\subfigure[Calculation of $Q, K, V$]{
\begin{minipage}[t]{0.65\textwidth}
\centering
\includegraphics[scale=0.45]{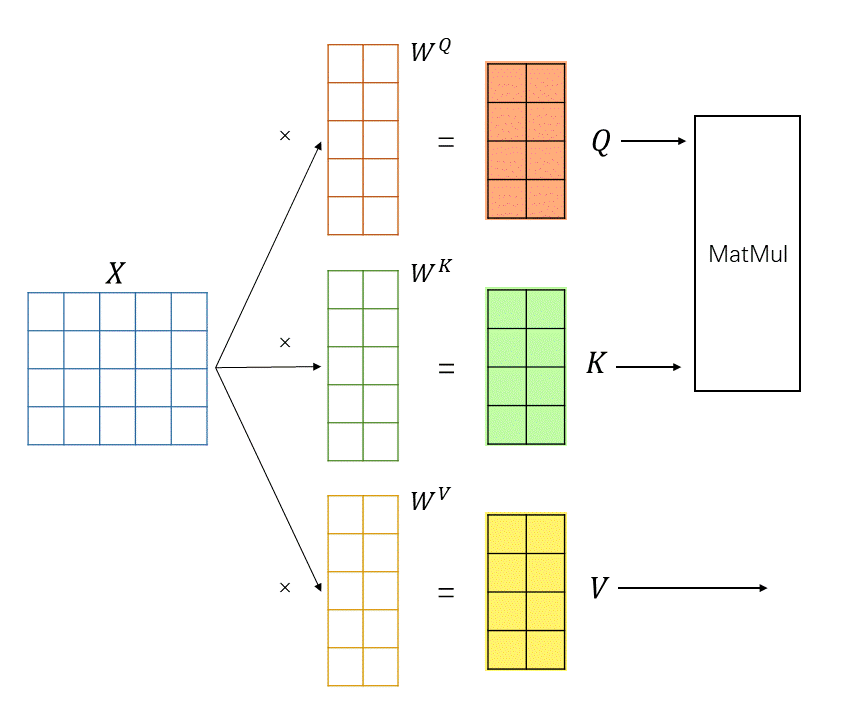}
\label{fig:QKV}
\end{minipage}
}
\centering
\caption{Transformer Architecture}
\label{fig：Transformer}
\end{figure}

The construction of LLMs usually superimposes many encoders and decoders. The multi-head attention mechanism in the encoder or the decoder contains multiple self-attention. So a large model like GPT-3 may contains billions of parameters which are mainly consist of parameter matrices in the self-attention module.

\subsection{LoRA(Low-Rank Adaptation)}

Fine-tuning, a central topic in AI model optimization, is one of the most common practices in transfer learning. When working on small models containing fewer than a million parameters, conducting full-parameter fine-tuning is manageable and does not present significant challenges. However, with the emergence of the GPT family of natural language processing models, the trend of increasing model size has become more and more apparent, and language models with billion parameter scales have become one of the hot topics in today's development. In this context, full-parameter fine-tuning operations are not only more memory intensive but also slower. Comparatively speaking, Parameter-Efficient Fine-Tuning(PEFT) methods are becoming greater importance.

LoRA is a PEFT technique released by Microsoft Research in 2021\cite{LoRA} for adapting large models to specific tasks and datasets, it is one of the most versatile and at the same time the most effective PEFT methods available. It significantly reduces the number of model parameters that need to be trained by applying a low-rank decomposition on the weight matrices of the LLM, thus reducing computational complexity and memory requirements. This approach allows efficient fine-tuning with limited resources while maintaining model performance. Compared to Adam's fine-tuned GPT-3 175B, LoRA can reduce the number of trainable parameters by a factor of 10,000 and reduce GPU memory requirements by a factor of 3\cite{LoRA}.

Many previous work has shown that over-parameterized large models have lower intrinsic dimensionality. The main idea behind LoRA is that changes in weights during model fine-tuning also have lower intrinsic dimensionality. Specifically, if $W_{d\times d}$ represents the weights of a single layer and $\Delta W_{d\times d}$ represents the change in weights during model adaptation, LoRA means that $\Delta W_{d\times d}$ is a low-rank matrix, i.e:$$rank(\Delta W_{d\times d})\ll d$$

\textbf{Low-rank decomposition. }Given $W\in \mathbb{R}^{d\times d}$ as the pre-trained weight, and $\Delta W\in \mathbb{R}^{d\times d}$ as the finetune incremental weight. Denote the input as $x$ and the output as $h$, then we have: $h=Wx+\Delta Wx$. In LoRA, we approximate $\Delta W$ with matrices A and B.

\begin{itemize}
\item $A\in \mathbb{R}^{d\times r}$: The low-rank matrix A, where r is the ``rank", uses a random Gaussian initialization.
\item $B\in \mathbb{R}^{r\times d}$: The low-rank matrix B  initializes to a zero matrix.
\end{itemize}

After the above splitting, $\Delta W$ is rewritten in the form of $\Delta W=AB$, which makes the size of fine-tuning parameter matrix reduced from $d*d$ to $2*r*d\,(r\ll d)$ without changing the dimension of the output data, i.e., $$h=Wx+ABx.$$

In addition, for the two low-rank matrices, an adjustment will be done with the hyperparameter $\alpha$ (a constant), which is used as the scaling rate to multiply directly with the low-rank matrices, i.e., the final output is: $$h=Wx+\frac{\alpha}{r}ABx.$$

\textbf{Applied to Transformer.}
LoRA can be applied to any subset of the weight matrices in a neural network to reduce the number of trainable parameters. In Transformer, there are four weight matrices $(W^Q, W^K, W^V, W^O)$ in the self-attention and two in the Multilayer Perceptron(MLP) module. For simplicity and parameter efficiency, we freeze the MLP module and focus only on the self-attention part. We treat $W^Q$(or $W^K$, $W^V$) as a single matrix of dimension $d\times d$, add bypass matrices to it and adjust the weights according to the downstream tasks, as shown in Fig.\ref{fig:LoRA}(Take $W^Q$ for example).

\begin{figure}
\centering
\includegraphics[scale=0.6]{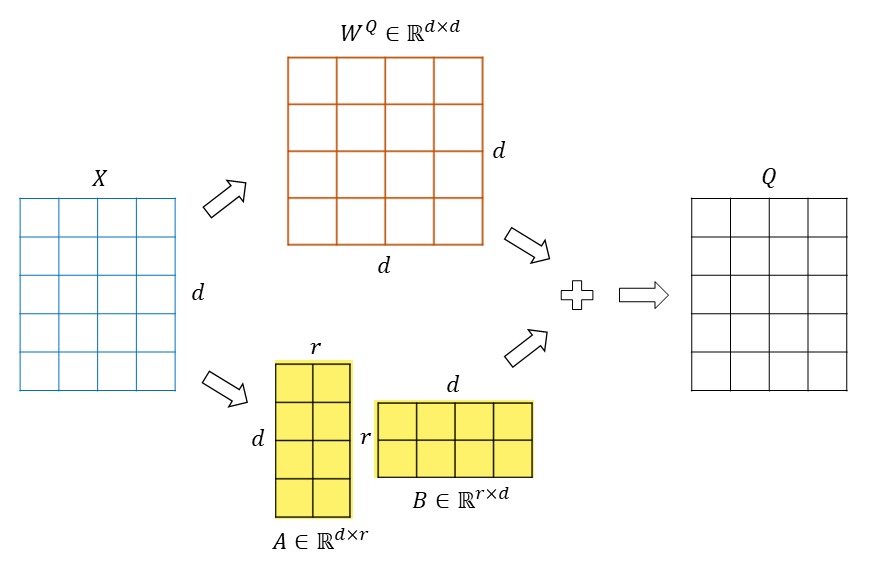}
\caption{LoRA applied to transformer with $d = 4, r = 2$}
\label{fig:LoRA}
\end{figure}

\subsection{Homomorphic Encryption}

Homomorphic encryption is a cryptographic technique based on the computational complexity theory of mathematical problems. The idea was first proposed by R. Rivest et al.\cite{1978R} in 1978. The term ``homomorphic" means that after homomorphic encryption, the result of a computation on the ciphertext is the same as the result of a direct computation on the plaintext after decryption. Let $F$ be a function that does computation on $x$ and $y$, $E$ be the encryption function and $D$ be the decryption function, then we have:$$D(F(E(x),E(y)))=F(x,y).$$

Homomorphic encryption methods are divided into three categories: Partially Homomorphic Encryption(PHE), Somewhat Homomorphic Encryption(SHE) and Fully Homomorphic Encryption(FHE). PHE supports partial forms of computation on the ciphertext, such as addition only or multiplication only. Those that only support addition are called additive homomorphic encryption algorithms, and those that only support multiplication are called multiplicative homomorphic encryption algorithms. SHE only supports a limited number of additive and multiplicative operations on the ciphertext (if the number of operations is too many, then the ciphertext can not be correctly decrypted because of too much noise). FHE supports arbitrary computation on the ciphertext, including both addition and multiplication operations. According to the logical completeness of the algorithm, a homomorphic algorithm that supports both ciphertext multiplication and ciphertext addition can support arbitrary computations on the ciphertext. 

In 2009, Gentry\cite{Gentry} proposed a FHE scheme using ideal lattices, which shows the first theoretically feasible blueprint of FHE. Since then, three generations of technical systems have been developed for FHE.

The CKKS FHE algorithm, proposed by Cheon et al.\cite{CKKS} in 2017, is a second-generation algorithm that is able to support both ciphertext addition and ciphertext multiplication, which makes it possible to apply high-precision approximations to arbitrary functions when calculating the gradient using the CKKS algorithm. In addition, the CKKS algorithm employs a rescaling technique to keep the message sizes before and after encoding basically unchanged during the homomorphic computation. This method ensures that the maximum ciphertext modulus required by the scheme grows linearly with the depth of the arithmetic circuits, and greatly improves the efficiency of the scheme. The CKKS algorithm is currently the most suitable algorithm for numerical computation, and can be implemented using the Microsoft SEAL library\cite{sealcrypto} developed by Microsoft's Cryptography and Privacy Research Group.

\subsection{Learning with Errors }

Learning with Errors (LWE) problem was introduced by Regev in 2005\cite{LWE} and has been widely used in the construction of cryptographic schemes. LWE is proven to be hard since  there is a reduction from LWE to lattice-based hard problems. There are two versions of LWE problems: search-version and decisional version,  which are defined as follows.

\begin{definition}[Search-version LWE(LWE)]
Given $m$ samples following the LWE distribution, that is $\{(A\in\mathbb{Z}_q^{m\times n}, b=As+e\,\mathrm{mod}\,q)\}$, where $A\sim U(\mathbb{Z}_q^{m\times n})$ is randomly selected and $s,e\sim \chi$. The goal of search-LWE is to find $s$.
\end{definition}

\begin{definition}[Decisional LWE(DLWE)]
Given $m$ samples, distinguish whether they follow the LWE distribution $\{(A\in\mathbb{Z}_q^{m\times n}, b=As+e\,\mathrm{mod}\,q)\}$ or uniformly distribution  $\{(A\in\mathbb{Z}_q^{m\times n}, b\in U(\mathbb{Z}_q^n)\}$.
\end{definition}
Many variants have been explored since the introduction of LWE. The Continuous LWE(CLWE) problem was proposed by Bruna et al. in 2020\cite{CLWE} and can be seen as a continuous variant of the LWE problem. Bruna et al. also prove that CLWE can be reduced to LWE under certain conditions. 
\begin{definition}[CLWE]
Let $n$ be an integer, and  $\gamma\geq2\sqrt{n}$, $\beta\in(0,1)$. Given $m$ sample $\{(A\in R^{m\times n}, b=\gamma As+e\,\mathrm{mod}\,1 \in R^m)\}$, where the secret vector  $s\in R^n$ is of length 1,
the elements in $A$ follow the standard normal distribution, i.e. $N(0,1)$ and $e$ follows $N(0,\beta)$. The CLWE problem is to find the secret vector $s$.
\end{definition}

\section{Practical Secure Inference Algorithm for Fine-tuned Large Language Model Based on FHE}

\subsection{Open-LLM + Private-LoRA}

Depending on whether private data is used in training, we split the weights involved in the inference process of the LLMs in the form of modules and propose an ``Open-LLM + Private-LoRA" structure. `Open-LLM' means the open-sourced pre-trained LLM, which provides the basic semantic comprehension ability, and the relevant model parameters are all publicly available. `Private-LoRA' is the rank-decomposition bypass matrix trained by private data, added to the base model to achieve the fine-tuning effects, such as extra knowledge in specialized areas.

By decomposing this structure, it can be seen that the `Open-LLM' does not involve private data and can be given to the user, who can locally compute it in plaintext. The `Private-LoRA' is trained using private data, the parameters of the rank-decomposition matrix need to be kept confidential, so it is saved on the server where it is trained. Thus users need to interact with the server to complete the inference. To protect users' input from leakage, users will encrypt their input before interacting, and the calculations performed at the server side are all in ciphertext form. According to the above method, the complete model inference process can be decomposed into two parts, one of which is computed using the plaintext only and the other is computation involves with ciphertext. So as to minimize the number of weights involved in the ciphertext computation to improve the efficiency of the  inference algorithm. The specific structural decomposition and encryption part are schematically shown in Fig.\ref{construct_fig}.

\begin{figure}
\centering
\includegraphics[scale=0.6]{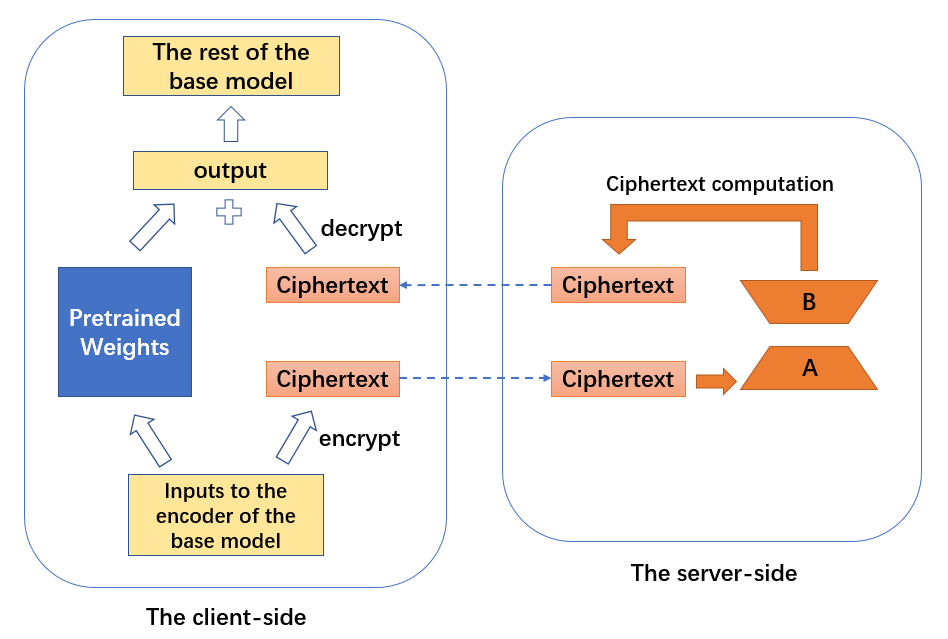}
\caption{data transmission between client and server in the self-attention of Transformer}
\label{construct_fig}
\end{figure}

The initialize process is shown in \ref{construction}. $M$ represents the base LLM, $FineTune(M)$ denotes the LoRA fine-tuning of the model $M$, and $L$ denotes the low-rank matrix part obtained by the fine-tuning.

\begin{figure}[ht]
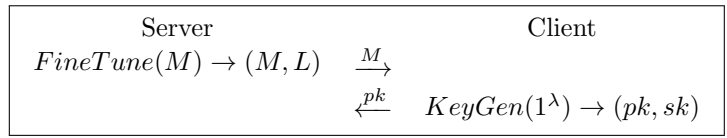
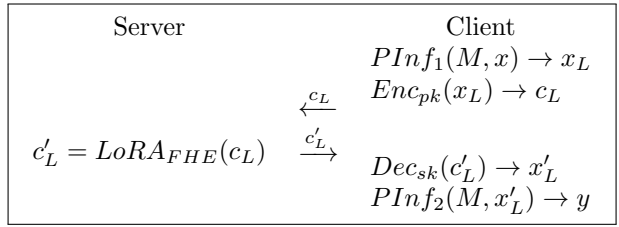

    \centering
    \subfigure[Initialize Process]{
    \begin{minipage}[b]{1\textwidth} 
        $${\boxed{
        \begin{tabular}{ccc}
        Server&&Client\\
        \makecell[l]{
        $FineTune(M)\rightarrow(M, L)$\\
        \vspace*{5pt}
        }
        &
        \makecell{
        $ \xrightarrow{M}$\\
        $\xleftarrow{pk}$
        }
        &
        \makecell[l]{
        \vspace*{5pt}\\
        $KeyGen(1^\lambda)\rightarrow(pk, sk)$
        }
        \end{tabular}
        }}$$
        \label{construction}
    \end{minipage}
    }
    \qquad 
    \subfigure[Inference Process]{
    \begin{minipage}[b]{0.45\textwidth} 
        $${\boxed{
        \begin{tabular}{ccc}
            Server&&Client\\
            \makecell[l]{
            \vspace*{5pt}\\
            \vspace*{5pt}\\
            $c_L'=LoRA_{FHE}(c_L)$ \\
            \vspace*{5pt}\\
            }
            &
            \makecell{
            \vspace*{5pt}\\
            $\xleftarrow{c_L}$\\
            $\xrightarrow{c_L'}$\\
            \vspace*{5pt}
            }
            &
            \makecell[l]{
            $PInf_1(M, x)\rightarrow x_L$\\
            $Enc_{pk}(x_L)\rightarrow c_L$\\
            \vspace*{5pt}\\
            $Dec_{sk}(c_L')\rightarrow x_L'$\\
            $PInf_2(M, x_L')\rightarrow y$
            }
            \end{tabular}
            }}$$
            \label{Ouralgorithm}
    \end{minipage}
    }
\caption{``Open-LLM + Private-LoRA” Structure}
\end{figure}

During the inference process, when a calculation  needs to go through the fine-tuning part, the specific interaction flow is shown in \ref{Ouralgorithm}. The user input data is denoted by $x$, $PInf_1(M, x)$ and $PInf_2(M, x)$ denote the computation of the base model $M$ on $x$ before and after the fine-tuning part, respectively. $LoRA_{FHE}()$ denotes the computation of $c_L$ under FHE.

\subsection{Private Linear Layer}

In this section, we propose a general method for transforming any linear layer into another one that enhances security against model extraction attacks while preserving its functionality. Firstly, we outline the problem.

\textbf{Attack for the plain linear layer parameter.} The weights of a linear layer can be viewed as a  matrix $A \in \mathbb{R}^{m\times n}$, then for an input $X \in \mathbb{R}^{d\times m}$, it outputs $Y=XA\in \mathbb{R}^{d\times n}$. Consider an adversary model that attempts to obtain the weights of linear layer through multiple queries, without loss of generality, we discuss the case where $d = 1$ for simplicity.
Assume that the adversary could choose any input and get the corresponding result, then just by inputting $x_1=(1,0,\ldots,0),\ldots,x_n=(0,0,\ldots,1)$,  we have an attack of complexity $O(n)$ that getting all parameters of $A$. Therefore it is easily seen that a plain linear layer is naturally weak against model extraction attack. 

\textbf{Our proposed method.} To avoid this problem, our method can be divided into two steps, in the first step, one should replace the structure of the linear layer by the following equation:  

$$y = 
\begin{pmatrix}
x & {x'}
\end{pmatrix}\begin{pmatrix}
{A} \\
{E'}
\end{pmatrix} + sA\quad \mathrm{mod}\quad  q,$$
where $x'\in \mathbb{R}^{d\times m'}, E' \in \mathbb{R}^{m'\times n}, s\in \mathbb{R}^{d\times m}, q \in   \mathbb{R} $.

More specially, $x'$ is fulfilled with $1$ in every entries,  $s$ consists of $d$ vectors which are separately sampled from a uniform distribution over all  vectors with the length $\gamma$ in $\mathbb{R}^{ m}$ and $q>0$ is a fixed real number parameter.

Then we need to train this  network to make it convergent by choosing  $\gamma, q>0$ and initializing $E'$ with a small Gaussian distribution for each entry and keep $x', s, q$ fixed through the training process. 

Once the training is complete, the second step corresponds to the inference process. With $x$ as the input, the output $y$ is computed as: 

$$y = xA + x'(E' \odot P)  + sA + kq,$$ where $k \in \mathbb{Z}^{d\times n}$ is a random matrix whose entries are separately sampled from a uniform distribution in $[-\lfloor q\rceil, \lfloor q \rceil ]$,  $P \in \{0,1\}^{m' \times n}$ is a random matrix whose entries are separately sampled from a Bernoulli distribution with the success probability $p$ and $\odot$ represents scalar multiplication of matrix.

Since $x'$ are fulfilled with $1$ and $P$ is a random matrix whose entries are separately sampled from a Bernoulli distribution, we have $E = x'(E' \odot P ) \in \mathbb{R}^{d\times n}$ and let $e_{i,j}$ denote the $i$-th  row and $j$-th column of $E$ and $e'_{i,j}$ denote the entry at $i$-th row and $j$-th column of $E'$, then:

$$e_{i,j} = \sum_{i=1}^mp_{i,j}e_{i,j}',$$
where $p_{i,j}$ is sampled from a Bernoulli distribution and $e_{i,j}$ is a determined weights, so $p_{i,j}e_{i,j}$ are independent bounded variables and their sum $e_{i,j}$ can be viewed as a Gaussian distribution according to the Hoeffding's inequality. Besides, introducing random variables by $P$ will not bring significant impact compared with the original matrix in neural network because this corresponds to a traditional method called dropout. Dropout is a regularization technique used in neural networks to prevent overfitting\cite{Dropout}. During a training process, dropout method randomly sets a fraction of the neurons to zero at each iteration with fixed probability.

As  $k \in \mathbb{Z}^{d\times n}$ is a random matrix whose entries are separately sampled from a uniform distribution in $[-\lfloor q\rceil, \lfloor q \rceil ]$,  when receiving $y$, a determined result can be obtained by:

$$y = xA + E + sA \, \mathrm{mod} \, q. $$

Now let us discuss its security against model extraction attack, and the problem can be defined as  follows. To simplify the description, we will take $d = 1$ without losing generality:

\begin{definition}[$SolveMatrix(t,m,n, q, \gamma,\beta)$]
Given $t=poly(m,n)$ samples $\{(x_i\in\mathbb{R}^m, b_i=(x_i+s)A+e_i\,\mathrm{mod}\,q)\}_{i=1,2,\cdots,t}$, where $s\sim \chi^{'}(\gamma)$ as a secret vector of length $\gamma$, $e_i\sim \chi(\beta)$ as a  vector of length $\beta$. $A$ is a random matrix  with each element following a normal distribution with a standard deviation of $1$ and $x_i$ is a random vector. Solve to obtain $A$.
\end{definition}

\begin{theorem}
\label{the:LWE}
For parameters  $\gamma/q\ge2\sqrt{m}$, if there exists an adversary that can solve the $SolveMatrix(t,m,n,$ $q,\gamma,\beta)$ problem, then it can distinguish the LWE problem, i.e., solve $CLWE(m,n,\gamma/q ,\beta/q)$.
\end{theorem}

\begin{proof} 
In order to prove the theorem, we need to convert instances of the CLWE problem into instances of the $SolveMatrix(t,m,n,q, \gamma,\beta)$ problem. For the $t=poly(m,n)$ samples input of CLWE problem,  algorithm \ref{SecuritProof-CLWE} will output corresponding $SolveMatrix$ samples. Then we pass the returned samples $(x_{i}, b'_{i})$ to the adversary  to get the corresponding matrix.

\begin{figure}
\centering
\includegraphics[scale=0.5]{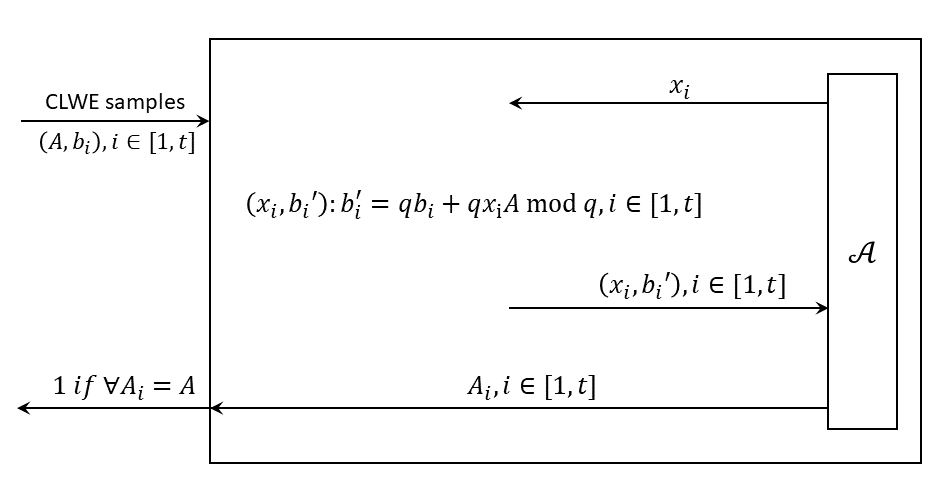}
\caption{Security Reduction Process}
\label{CLWE}
\end{figure}

\begin{algorithm}[!ht]
    \caption{Convert instances of  CLWE problem into instances of  $SolveMatrix(t=poly(m,n),m,n,q,\gamma,\beta)$ problem.}
    \label{SecuritProof-CLWE}
    \SetAlgoLined
    \SetKwInOut{Input}{Input}
    \SetKwInOut{Output}{Output}

    \Input {$t=poly(m,n)$ samples of the CLWE problem inputs$\{(A\in\mathbb{R}_1^{m\times n}, b\in\mathbb{R}_1^n\}$}
    \Output {$t$ samples $\{x_{i},b'_{i}\}$}
 
    \For{$i \leftarrow 1$ \KwTo$ t $}{
        \textbf{The Adversary:} Randomly pick $x_{i}\in \mathbb{R}^m$
        \\Compute $b'_{i}=qb_i+qx_{i}A\,\mathrm{mod}\,q$
    }
    \Return{$(x_{i},b'_{i})_{i=1,\cdots,t}$}  
\end{algorithm}

Depending on the output, we consider two cases:
\begin{itemize}
    \item If the input to the CLWE problem is a sample that follows the CLWE distribution, i.e., $b_i=\gamma/q  \cdot sA+e_i\,\mathrm{mod}\,1$. It can be inferred that for sample $\{x_{i},b'_{i}\}$, there is $b'_{i}=s^{'}A+x_{i}(qA)+qe_i\,\mathrm{mod}\,q$  where $s^{'}=\gamma s$. At this time, $(x_{i},b'_{i})$ is a sample input of the problem $SolveMatrix(t,m,n,q,\gamma,\beta)$. Given that adversary   has the ability to solve the $SolveMatrix$, the adversary can solve the problem  and return   $  A $.
    \item If the input samples to the CLWE problem are samples following a uniform distribution, then in the view of the adversary, the obtained sample $b_{i}$ follows uniform distribution, which does not meet the definition of SolveMatrix problem. Since $A$ follows Gaussian distribution with a standard deviation of $1$. Therefore, the probability of the adversary returning the specific   $ A  $ is negligible.
\end{itemize}

By determining whether the return value of the adversary is the input matrix $ A  $, we can distinguish whether a sample follows CLWE distribution or uniform distribution. This indicates an algorithm that  can solve  CLWE problem. This completes the proof of the theorem \ref{the:LWE}.
\end{proof}
It should be noted that this theorem demands the parameters satisfy   $\gamma/q\ge2\sqrt{m}$ to complete the proof. Besides, to ensure the  updated private linear layer preserves its original functionality, we need to choose  $\gamma$ and $q$  to make the network convergent during the retraining process. As the  convergent of a neural network highly depends on the specific training task, it is not sure that we can find such $\gamma$ and $q$ for every specific task. However, we can still finish the proof by reducing the problem to a variant of CLWE where $\gamma \ge \alpha \sqrt{n}$ for some $\alpha \in(0,2)$, whose  hardness still needs further studies.

\subsection{Implementation of PLL based on   FHE}

By utilizing PLL technology, we can protect the weights on the server side from leakage. To further achieve the security of users' input, we use leveled fully homomorphic algorithm in the computation of PLL, e.g. the inference process involving LoRA rank-decomposition matrix operations. Since the packed encoding technique in the hierarchical fully homomorphic algorithm requires each operation to be performed on every element in the ciphertext vector, it is necessary to transform the matrix-based plaintext computation process in LoRA reasoning into a vector-based ciphertext inference process. According to the formula $y = xA + x'(E' \odot P)  + sA + kq,$ where $x\in \mathbb{R}^{d\times m}, A\in \mathbb{R}^{m\times n}, q\in\mathbb{R} $. The process is as follows:

Given the input matrix $x$. According to the structure of LoRA, $A=A_1A_2$, where $A_1\in \mathbb{R}^{m\times r}, A_2\in \mathbb{R}^{r\times n}$. the low-rank matrices $A_1$ and $A_2$ are described as follows:

$x~ = ~\begin{bmatrix}
x_{11} & x_{12} & \cdots & x_{1m} \\
x_{21} & x_{22} & \cdots & x_{2m} \\
 \vdots & \vdots & \ddots & \vdots \\
x_{d1} & x_{d2} & \cdots & x_{dm}
\end{bmatrix}, $
$A_1~ = ~\begin{bmatrix}
a_{11} & a_{12} & \cdots & a_{1r} \\
a_{21} & a_{22} & \cdots & a_{2r} \\
 \vdots & \vdots & \ddots & \vdots \\
a_{m1} & a_{m2} & \cdots & a_{mr}
\end{bmatrix}, $
$A_2~ = ~\begin{bmatrix}
b_{11} & b_{12} & \cdots & b_{1n} \\
b_{21} & b_{22} & \cdots & b_{2n} \\
 \vdots & \vdots & \ddots & \vdots \\
b_{r1} & b_{r2} & \cdots & b_{rn}
\end{bmatrix}.$

According to Section 3.2, $x'\in \mathbb{R}^{d\times m'}$,$E'\in \mathbb{R}^{m'\times n}$,$s\in\mathbb{R}^{d\times m}$ are fixed matrices, $P\in\{0,1\}^{m'\times n},k\in\mathbb{Z}^{d\times n}$ are random matrices. Since the $Q = x'(E' \odot P)  + sA + kq$ part is calculated in plaintext on the server, in this section, we assume that the randomly obtained matrix $Q$ in round $t$ is a plaintext matrix $Q_t = x'(E' \odot P_t)  + sA + k_tq $ as 
$Q_t~ = ~\begin{bmatrix}
q(t)_{11} & q(t)_{12} & \cdots & q(t)_{1n} \\
q(t)_{21} & q(t)_{22} & \cdots & q(t)_{2n} \\
 \vdots & \vdots & \ddots & \vdots \\
q(t)_{d1} & q(t)_{d2} & \cdots & q(t)_{dn}
\end{bmatrix}.$

The scheme's steps are as follows:

\textbf{Step 1:} The client encrypts the input matrix $x$ and transmits the encrypted matrix $ct_x$ to the server:
$${ct}_{x}~ = ~Enc\begin{bmatrix}
x_{11} & x_{12} & \cdots & x_{1m} \\
x_{21} & x_{22} & \cdots & x_{2m} \\
 \vdots & \vdots & \ddots & \vdots \\
x_{d1} & x_{d2} & \cdots & x_{dm}
\end{bmatrix}.$$

\textbf{Step 2:} The server receives the ciphertext input matrix $ct_x$, does the ciphertext-plaintext multiplication operation with the low-rank matrix $A_1$, and rescales the resulting result in $p$ bits:
$${ct}_{x1}~\leftarrow~ReScale\left( {Mult\left( {A_1;{ct}_{x}} \right);p} \right).$$

The following shift-and-sum operation is performed on $ct_{x1}$ to obtain $ct_{x3}$, where $j=0,1,...,log(m+1)-1$ and * denotes an irrelevant parameter:
$$\left. {ct}_{x3}~\leftarrow~Add\left( {{ct}_{x2},~Rotate\left( {{ct}_{x2};2^{j}} \right)} \right) \right.,$$
$${ct}_{x3}~ = Enc\begin{bmatrix}
{\sum\limits_{k = 1}^{m}{x_{1,k}\cdot a_{k,1}}} & \mathbf{*} & \cdots & \mathbf{*} \\
{\sum\limits_{k = 1}^{m}{x_{1,k}\cdot a_{k,2}}} & \mathbf{*} & \cdots & \mathbf{*} \\
 \vdots & \vdots & \ddots & \vdots \\
{\sum\limits_{k = 1}^{m}{x_{1,k}\cdot a_{k,r}}} & \mathbf{*} & \cdots & \mathbf{*} \\
 \vdots & \vdots & \ddots & \vdots \\
{\sum\limits_{k = 1}^{m}{x_{d,k}\cdot a_{k,1}}} & \mathbf{*} & \cdots & \mathbf{*} \\
{\sum\limits_{k = 1}^{m}{x_{d,k}\cdot a_{k,2}}} & \mathbf{*} & \cdots & \mathbf{*} \\
 \vdots & \vdots & \ddots & \vdots \\
{\sum\limits_{k = 1}^{m}{x_{d,k}\cdot a_{k,r}}} & \mathbf{*} & \cdots & \mathbf{*}
\end{bmatrix}.$$

\textbf{Step 3:} Divide the irrelevant parameters and set the plaintxt matrix D:
$$D = \begin{bmatrix}
1 & \cdots & 0 \\
1 & \cdots & 0 \\
 \vdots & \ddots & \vdots \\
1 & \cdots & 0
\end{bmatrix}.$$

Do ciphertext-plaintext multiplication of $ct_{x3}$ and D to get $ct_{x4}$:
$${ct}_{x4} ~\leftarrow~ ReScale\left( {CMult\left( {D;{ct}_{x3}} \right);p} \right), $$
$${ct}_{x4}~ = Enc\begin{bmatrix}
{\sum\limits_{k = 1}^{m}{x_{1,k}\cdot a_{k,1}}} & 0 & \cdots & 0 \\
{\sum\limits_{k = 1}^{m}{x_{1,k}\cdot a_{k,2}}} & 0 & \cdots & 0 \\
 \vdots & \vdots & \ddots & \vdots \\
{\sum\limits_{k = 1}^{m}{x_{1,k}\cdot a_{k,r}}} & 0 & \cdots & 0 \\
 \vdots & \vdots & \ddots & \vdots \\
{\sum\limits_{k = 1}^{m}{x_{d,k}\cdot a_{k,1}}} & 0 & \cdots & 0 \\
{\sum\limits_{k = 1}^{m}{x_{d,k}\cdot a_{k,2}}} & 0 & \cdots & 0 \\
 \vdots & \vdots & \ddots & \vdots \\
{\sum\limits_{k = 1}^{m}{x_{d,k}\cdot a_{k,r}}} & 0 & \cdots & 0
\end{bmatrix}.$$

Replicate the inner product values to other columns to obtain $ct_{x5}$, where $j=0,1,\dots,log(m+1)-1$:

$$\left. {ct}_{x5}~\leftarrow~Add\left( {{ct}_{x4},~Rotate\left( {{ct}_{x4};-2^{j}} \right)} \right) \right.,$$

$${ct}_{x5}~ = Enc\begin{bmatrix}
{\sum\limits_{k = 1}^{m}{x_{1,k}\cdot a_{k,1}}} & {\sum\limits_{k = 1}^{m}{x_{1,k}\cdot a_{k,1}}} & \cdots & {\sum\limits_{k = 1}^{m}{x_{1,k}\cdot a_{k,1}}} \\
{\sum\limits_{k = 1}^{m}{x_{1,k}\cdot a_{k,2}}} & {\sum\limits_{k = 1}^{m}{x_{1,k}\cdot a_{k,2}}} & \cdots & {\sum\limits_{k = 1}^{m}{x_{1,k}\cdot a_{k,2}}} \\
 \vdots & \vdots & \ddots & \vdots \\
{\sum\limits_{k = 1}^{m}{x_{1,k}\cdot a_{k,r}}} & {\sum\limits_{k = 1}^{m}{x_{1,k}\cdot a_{k,r}}} & \cdots & {\sum\limits_{k = 1}^{m}{x_{1,k}\cdot a_{k,r}}} \\
 \vdots & \vdots & \ddots & \vdots \\
{\sum\limits_{k = 1}^{m}{x_{d,k}\cdot a_{k,1}}} & {\sum\limits_{k = 1}^{m}{x_{d,k}\cdot a_{k,1}}} & \cdots & {\sum\limits_{k = 1}^{m}{x_{d,k}\cdot a_{k,1}}} \\
{\sum\limits_{k = 1}^{m}{x_{d,k}\cdot a_{k,2}}} & {\sum\limits_{k = 1}^{m}{x_{d,k}\cdot a_{k,2}}} & \cdots & {\sum\limits_{k = 1}^{m}{x_{d,k}\cdot a_{k,2}}} \\
 \vdots & \vdots & \ddots & \vdots \\
{\sum\limits_{k = 1}^{m}{x_{d,k}\cdot a_{k,r}}} & {\sum\limits_{k = 1}^{m}{x_{d,k}\cdot a_{k,r}}} & \cdots & {\sum\limits_{k = 1}^{m}{x_{d,k}\cdot a_{k,r}}}
\end{bmatrix}.$$

Perform a multiplication operation on the ciphertexts $ct_{x5}$ and $A_2$ and rescale the resulting result by $p$ bits:
$${ct}_{x6}~\leftarrow~ReScale\left( {Mult\left( {A_2;{ct}_{x5}} \right);p} \right).$$

Sum $ct_{x6}$ every r lines to get the $ct_{x7}$, where $k=0,1,\dots,n-1$:
$$\left. i = \lbrack kr+1, kr + r\rbrack, {ct}_{x7}[k]~\leftarrow~Add\left( {{ct}_{x6}[i+1],~{ct}_{x6}[i]} \right) \right.,$$
$${ct}_{x7}~ = ~Enc\begin{bmatrix}
{\sum\limits_{i = 1}^{r}\left( {\sum\limits_{k = 1}^{m}{x_{1,k}\cdot a_{k,i} \cdot b_{i,1}}} \right)} & {\sum\limits_{i = 1}^{r}\left( {\sum\limits_{k = 1}^{m}{x_{1,k}\cdot a_{k,i} \cdot b_{i,2}}} \right)} & \cdots & {\sum\limits_{i = 1}^{r}\left( {\sum\limits_{k = 1}^{m}{x_{1,k}\cdot a_{k,i} \cdot b_{i,n}}} \right)} \\
{\sum\limits_{i = 1}^{r}\left( {\sum\limits_{k = 1}^{m}{x_{2,k}\cdot a_{k,i} \cdot b_{i,1}}} \right)} & {\sum\limits_{i = 1}^{r}\left( {\sum\limits_{k = 1}^{m}{x_{2,k}\cdot a_{k,i} \cdot b_{i,2}}} \right)} & \cdots & {\sum\limits_{i = 1}^{r}\left( {\sum\limits_{k = 1}^{m}{x_{2,k}\cdot a_{k,i} \cdot b_{i,n}}} \right)} \\
 \vdots & \vdots & \ddots & \vdots \\
{\sum\limits_{i = 1}^{r}\left( {\sum\limits_{k = 1}^{m}{x_{d,k}\cdot a_{k,i} \cdot b_{i,1}}} \right)} & {\sum\limits_{i = 1}^{r}\left( {\sum\limits_{k = 1}^{m}{x_{d,k}\cdot a_{k,i} \cdot b_{i,2}}} \right)} & \cdots & {\sum\limits_{i = 1}^{r}\left( {\sum\limits_{k = 1}^{m}{x_{d,k}\cdot a_{k,i} \cdot b_{i,n}}} \right)}
\end{bmatrix}.$$

\textbf{Step 4:} Add other plaintext entries, which is sent to the client:

$${ct}_{x8}~\leftarrow~Add\left({ct}_{x7},Q_t \right),$$
$${ct}_{x7}~ = ~Enc\begin{bmatrix}
{\sum\limits_{i = 1}^{r}\left( {\sum\limits_{k = 1}^{m}{x_{1,k}\cdot a_{k,i} \cdot b_{i,1}}} \right)}+q(t)_{11} & \cdots & {\sum\limits_{i = 1}^{r}\left( {\sum\limits_{k = 1}^{m}{x_{1,k}\cdot a_{k,i} \cdot b_{i,n}}} \right)}+q(t)_{1n} \\
{\sum\limits_{i = 1}^{r}\left( {\sum\limits_{k = 1}^{m}{x_{2,k}\cdot a_{k,i} \cdot b_{i,1}}} \right)}+q(t)_{21} & \cdots & {\sum\limits_{i = 1}^{r}\left( {\sum\limits_{k = 1}^{m}{x_{2,k}\cdot a_{k,i} \cdot b_{i,n}}} \right)}+q(t)_{2n} \\
 \vdots & \ddots & \vdots \\
{\sum\limits_{i = 1}^{r}\left( {\sum\limits_{k = 1}^{m}{x_{d,k}\cdot a_{k,i} \cdot b_{i,1}}} \right)}+q(t)_{d1} & \cdots & {\sum\limits_{i = 1}^{r}\left( {\sum\limits_{k = 1}^{m}{x_{d,k}\cdot a_{k,i} \cdot b_{i,n}}} \right)}+q(t)_{dn}
\end{bmatrix}.$$

\textbf{Step 5:} The client receives the ciphertext result $ct_{x8}$ and decrypts it to get the result of multiplying the input matrix $x$ with the weights of server. Then continue the model's operations.

Since the parameters involved in ciphertext computation are reduced as much as possible by the inference structure, communication overhead accounts for the majority of time consumption, so it is important to reduce the number of communications during the design of the scheme. For each time the inference process of PLL is called, three times ciphertext multiplication by plaintext operations along with two times communications  will be executed. 

\section{Experiments}

\subsection{Experiment Settings}
\textbf{Hardware Setup.} We use two devices to separately act as client and server in the experiments. Device 1 has a hardware configuration as Intel(R) Xeon(R) Platinum 8375C CPU @ 2.90GHz with RAM 48G, NVIDIA RTX A6000 and the operating system is Ubuntu 22.04.3 LTS with Linux 6.5.0-28-generic. Device 2 has a hardware configuration as Intel(R) Xeon(R) Gold 6133 CPU @ 2.50GHz with RAM 24G, NVIDIA RTX 4090, and operating system is Ubuntu 22.04.4 LTS with Linux 6.5.0-21-generic. The network bandwidth between the two devices is 3 Mbps, thus the round trip time is  about 25.6ms on average.

\textbf{Models and fine-tuning methods.} We use ChatGLM2-6B as the base model, which is a Chinese-English bilingual dialog model released in June 2023 and open-sourced in July by Zhipu AI and Tsinghua KEG Lab\cite{du2022glm}. ChatGLM-6B is based on the General Language Model(GLM) architecture\cite{GLM} with 6.2 billion parameters. ChatGLM2-6B has more powerful performance, faster inference, and weights that are fully open to academic research. Our LoRA is implemented using LLaMA-Factory\cite{llamafactory} to add local fine-tuning effects, i.e., adding bypass matrices to the Q,K,V weight matrices in the self-attention mechanism of the base model.

\textbf{Implementation of FHE.} The FHE algorithm used in this paper is CKKS. The realization of our scheme uses two programming languages, C++ and Python for better efficiency. Thus we use Microsoft SEAL library\cite{sealcrypto} for CKKS algorithm in the C++ part, and SEAL-Python library(a python binding for the Microsoft SEAL library) in the Python part. The parameters of CKKS are set to satisfy that the max multiplication depth reaches 4.  

\subsection{Improvement of Computing and Communication Efficiency}
To further improve the efficiency of our scheme, we use parallelism to reduce time for ciphertext computations. However, since Python's threads are essentially OS-native threads, each thread needs to acquire a global interpreter lock when executing the code. Since Python programs have only one interpreter, different threads need to compete for the same lock in order to execute their respective code, making real parallel execution impossible. This limitation prevents Python's multi-threading from achieving real parallel computation. So we realize the server-side's algorithm by C++ and then use thread pooling to achieve high parallelism, which can increase the inference efficiency by 20\% to 40\%.

For ciphertext transmission, Python's built-in serialization may result in data loss when deserialized under C++. Therefore, we choose to communicate with Protocol Buffers (protobuf)\cite{protobuf}. Protobuf is a lightweight and efficient structured data storage format to serialize structured data, and is very suitable for data storage or Remote Procedure Call (RPC) data exchange format. This serialization format provides a more reliable serialization and deserialization mechanism to better protect the integrity and accuracy of the data.

\subsection{Application Scenario Setting}
This algorithm can be used in a variety of vertical fields because of the strong generalization capabilities brought by LLM, such as  classification of sensitive data. Classification of sensitive data is a data protection measure commonly used in the field of information security, the main content of which is to classify data into different levels according to its importance and sensitivity, and then implement different protection measures to ensure the confidentiality, integrity and availability of data. Sensitive data encompasses a wide variety of information including, but not limited to, personally identifiable information, financial records, health data, trade secrets and confidential government information. Exposure, tampering, or unauthorized access to this data can have a significant impact on individuals, organizations, or society as a whole. If one wants to use LLMs to improve efficiency and manipulate sensitive data, privacy-preserving techniques must be employed to secure these data. Therefore we use this as an example of an application scenario for our proposed scheme.

In this application, 'Input' is a message that needs to be classified and 'Output' is a symbol representing the category of this message. In the training process, we uses 5360 sets of non-repeated data to fine-tuning the LLM by LoRA. And 283 sets of non-repeated data are chosen as the test set. Then we update the model by our method to obtain a LLM that can classification sensitive data and protect both input and private weights at the same time. It is proved experimentally that the fine-tuned training under our method still converges, the training loss values are shown in Fig.\ref{loss}. And the result shows that, the updated model's accuracy reaches about 97\% which preserves its original functionality and the accuracy loss is consist with the case when using the dropout method in the reference process.

\begin{figure}
\centering
\includegraphics[scale=0.4]{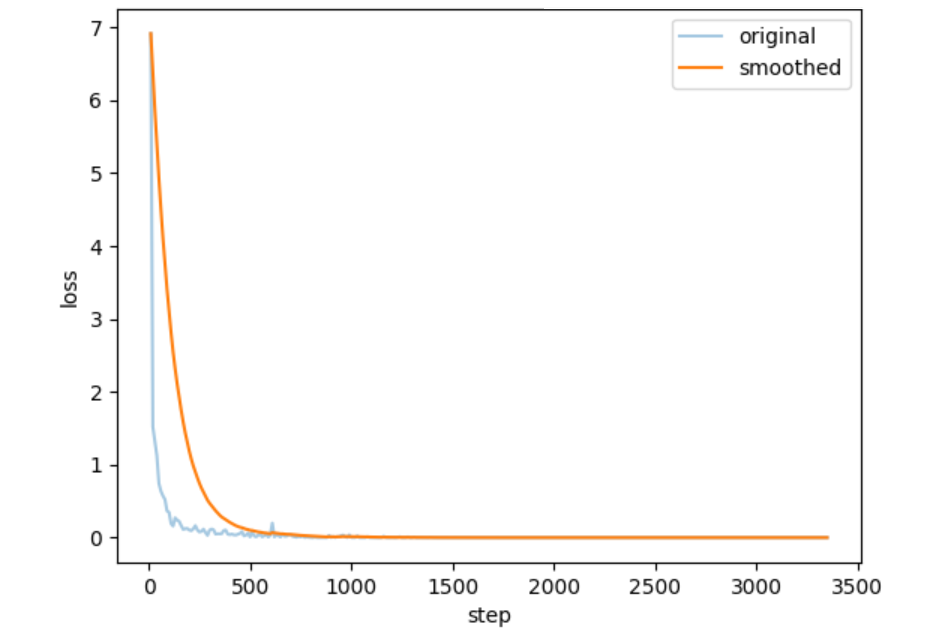}
\caption{loss values in training}
\label{loss}
\end{figure}

\subsection{Inference Speed}
Since the SIMD(Single Instruction Multiple Data)\cite{SIMD} in CKKS allows the homomorphic addition and multiplication to be operated in parallel, our method is good at dealing input of many tokens. The experiment demonstrates that when the number of tokens reaches 1000, the performance approaches the optimal level of 1.61 seconds per token where the rank in LoRA is set to 8. And the runtime with different token numbers is shown in Fig.\ref{fig:timeresult}. The efficiency of processing 500 tokens with different ranks in LoRA is shown in the Fig.\ref{fig:differentrank}. It can be seen that efficiency is approximately linearly correlated with the rank, which means that our method can be extended to linear layers with larger dimensions. As a result, our method also shows the potential to be applied for fine-tuned LLM where $r=d$. 

\begin{figure}
\centering
\begin{tikzpicture}
  \begin{axis}
    [
      xlabel=number of tokens,
      ylabel=time(s)/token,
      ymax=180
    ]  
    \addplot+[smooth]   
    coordinates
    {
      (50,315.586) (100,231.223) (200,188.6815)(500,171.5818) (700,168.355) (1000,160.6188)
    };
    
  \end{axis}
\end{tikzpicture}
\caption{Runtime with different number of tokens}
\label{fig:timeresult}
\end{figure}
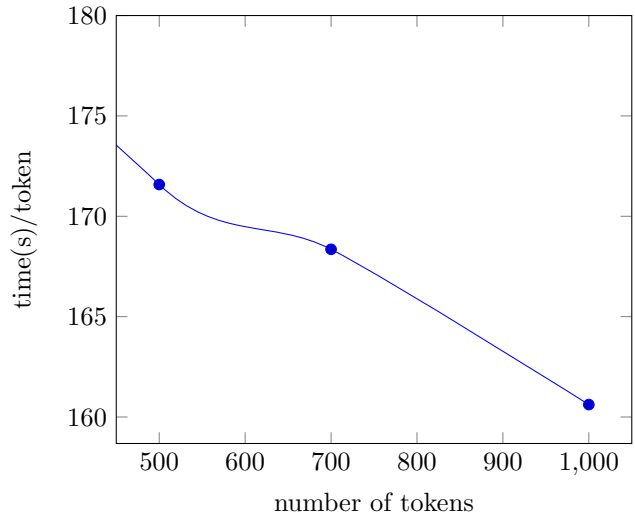

\begin{figure}
\centering
\begin{tikzpicture}
  \begin{axis}
    [
      xlabel=rank,
      ylabel=time(s)/token,
      ymax=7
    ]  
    \addplot+[smooth]   
    coordinates
    {
      (8,1.715818)(16,2.49392)(24,3.36596)(48,5.6657)
    };
    
  \end{axis}
\end{tikzpicture}
\caption{Runtime of 500 tokens with different rank}
\label{fig:differentrank}
\end{figure}
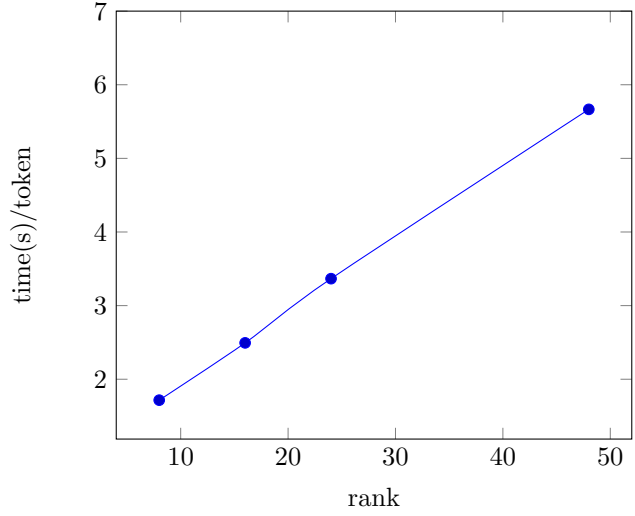

In Table \ref{table:table1} and \ref{table:table2}, we list the privacy inference algorithms for LLMs proposed in recent years, with the parameters of the LLMs and their efficiency. ``-" indicates that it is not explicitly stated in the original paper. We can see that the earlier proposed scheme is not suit for larger models. For large models with number of parameters more than a billion, PUMA is the most efficient method. And the comparison shows that our scheme takes about only 1.61s to generate 1 token in inference process, which is 127 times faster than PUMA. 

\begin{table}[]
\caption{Privacy inference algorithms implemented on smaller LLMs}  
\label{table:table1}  
\resizebox{\columnwidth}{!}{
\renewcommand{\arraystretch}{1.5}
\begin{tabular}{|c|c|c|c|p{1.3cm}|}
\hline
Time of publication & Schemes       & Experiments on    & quantity of parameters & Time           \\ \hline
2022              & THE-X\cite{THE-X}         & Bert-tiny         & $<$14.5M       & -              \\ \hline
2022              & Iron\cite{Iron}          & Bert-Large        & 340M                   & 6000s          \\ \hline
2023              & BumbleBee\cite{BumbleBee}       & GPT2-Base         & 117M                   & 204.6s         \\ \hline
2023              & CipherGPT\cite{CipherGPT}     & GPT2-Base         & 117M                   & 1500s          \\ \hline
2023              & PUMA\cite{PUMA}          & GPT2-Base         & 117M                   & 15.5s          \\ \hline
\end{tabular}}
\end{table}

\begin{table}[]
\caption{Privacy inference algorithms implemented on LLMs with billion+ parameter sizes}  
\label{table:table2}  
\resizebox{\columnwidth}{!}{
\renewcommand{\arraystretch}{1.5}
\begin{tabular}{|c|c|c|c|p{1.3cm}|}
\hline
Time of publication & Schemes       & Experiments on    & quantity of parameters & Time           \\ \hline
2023              & BumbleBee\cite{BumbleBee}       & LLaMA-7B          & 7B                     & 832.2s         \\ \hline
2023              & PUMA\cite{PUMA}          & LLaMA-7B          & 7B                     & 200s           \\ \hline
\textbf{2024}     & \textbf{Ours} & \textbf{ChatGLM2} & 6B                     & \textbf{1.61s} \\ \hline
\end{tabular}}
\end{table}

\section{Conclusion}
For the paradigm of fine-tuning a pre-trained LLM to improve the model's capability for a specialized domain, we combine privacy-preserving techniques include FHE and provable security theory, and propose an efficient and secure ciphertext inference scheme for LLMs that can protect the user-side inputs and the private training dataset on the server-side. Our method is to separate the low-rank matrix part obtained using LoRA fine-tuning from the pre-trained model, and only perform private computing in the former part to reduce the overhead. Besides, we also introduce a general method to transform a linear layer to PLL that preserves its original functionality and provides security against model extraction attacks. By using these technologies, the inference process of our scheme takes less than 1.61s for each token when the number of tokens is greater than 1000, which makes the scheme a practical privacy-preserving LLM. 

\bibliographystyle{alpha}
\bibliography{sample}

\end{document}